\documentclass{eptcs}
\providecommand{\event}{ADG 2021} 
\usepackage{graphicx}

\usepackage{url}

\usepackage{amssymb}
\usepackage{amsmath}
\usepackage{amsthm}
\usepackage{comment}
\usepackage{xcolor}
\usepackage{cancel}
\newcommand{\abs}{{\tt abs}}
\newcommand{\thee}{\mbox{\it TH}}

\newtheorem{theorem}{Theorem}
\newtheorem{lemma}{Lemma}
\newtheorem{remark}{Remark}

\def\orcidID#1{\unskip$^{\mbox{\href{https://orcid.org/#1}{\scriptsize{[#1]}} }}$}

\begin{document}

\title{Maximizing the Sum of the Distances between Four Points on the Unit Hemisphere\thanks{Supported by NSFC under Grant Nos. 12071282,  61772203.} 
}

\author{Zhenbing Zeng\orcidID{0000-0002-9728-1114} 
\institute{Shanghai University,  Department of Mathematics \\  Shanghai 200444, China}
\email{zbzeng@shu.edu.cn}
\and
Jian Lu\orcidID{0000-0002-0487-259X} 
\institute{Shanghai University,  Department of Mathematics \\  Shanghai 200444, China}
\email{lujian@picb.ac.cn} 
\and
\quad Yaochen Xu\orcidID{0000-0002-5039-2781} 
\institute{Shanghai University,  Department of Mathematics \\  Shanghai 200444, China}
\email{xuyaochen@sibcb.ac.cn}
\and
Yuzheng Wang\orcidID{0000-0002-0013-3466}
\institute{Shanghai University,  Department of Mathematics \\  Shanghai 200444, China}
\email{1290666930@qq.com}
}

\def\titlerunning{Maximizing the Sum of the Distances between Four Points on the Unit Hemisphere}
\def\authorrunning{Zeng, Lu, Xu \& Wang}

\maketitle              
\begin{abstract}
In this paper,  we prove a geometrical inequality which states that for any four points 
on a hemisphere with the unit radius,  the largest sum of distances between the points is $4+4\sqrt{2}$.
In our method,  we have constructed a rectangular neighborhood of the local maximum point in the feasible set,  
which size is explicitly determined, 
and proved that (1): the objective function is bounded by a quadratic polynomial which takes the local maximum point
as the unique critical point in the neighborhood,  and
(2): the rest part of the feasible set can be partitioned into a finite union of a large number of very small cubes 
so that on each small cube the conjecture can be verified by estimating the objective function with exact numerical computation. 

\end{abstract}
\section{Introduction}\label{introduction}
Assume that four points are placed on the hemisphere of the unit radius:
$$
S^2_{\geq 0}:=\{(x, y, z)|x^2+y^2+z^2=1,  z\geq 0\},  
$$ 
we want to find the largest value of the sum of distances between them. 
A similar problem for maximizing the sum of distances between $n$ points 
on the unit sphere has been studied by many people in past 
(see \cite{toth56},  \cite{stolarsky73},  and \cite{moser} for example),  where
we can see that the problem for four points is very easy and for $n\geq 5$ it becomes very difficult. 
To our knowledge,  the question for points on hemispheres seems not settled yet in the literature.

As we will prove in Lemma~\ref{lemma0},  if the sum of distances between four points is maximal,  
then the center of the sphere must lie in the interior or on one of the surfaces of the tetrahedron 
formed by the four points,  
which immediately implies that,  as shown in Fig.~\ref{fig-fourpointsonhemisphere},  
at least three of the four points must lie on the equator of the hemisphere: 
$$
S^1:=\{(x, y, 0)|x^2+y^2=1\}.
$$

\begin{figure}[!ht]
\begin{center}
\includegraphics[width=0.55\linewidth,  trim=0 30 0 0,  clip]{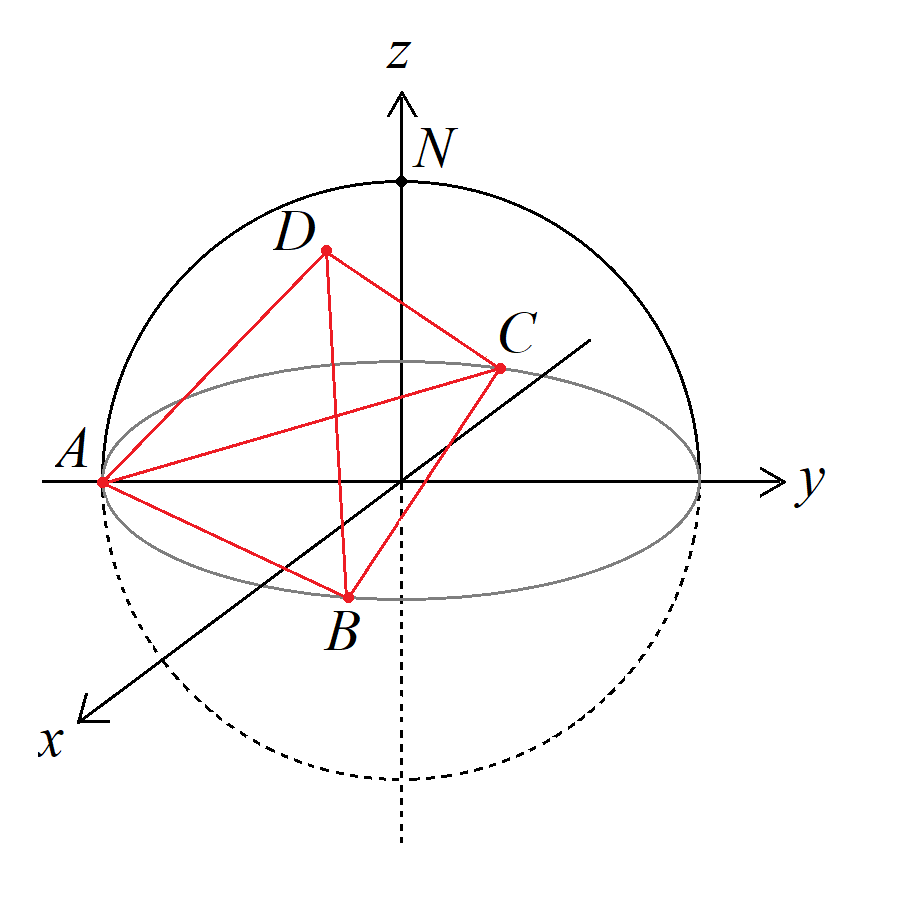}
\end{center}
\caption{Four points on a hemisphere,  three of them on the equator.}
\label{fig-fourpointsonhemisphere}
\end{figure}

\noindent  
In the beginning,  we guessed that the optimal configuration 
is formed by three vertices of an equilateral triangle inscribed in the equator 
with the fourth point at the North Pole,  so 
the largest sum is equal to $3\sqrt{3}+3\sqrt{2}\approx 9.43879311\cdots$. 
Soon we found that $3\sqrt{3}+3\sqrt{2}$ is not the maximal,  since if we put all four points on the equator in a regular square form,  
then the sum of distances is:
$$
4\sqrt{2}+4=9.65685424\cdots > 3\sqrt{3}+3\sqrt{2}\approx 9.43879311\cdots.
$$
Indeed,  it is easy to see that if $A, B, C$ are arbitrary points on the equator 
and the point $D\in S^2_{\geq 0}$ is located very close to the North Pole $N$,  
namely,  if the distance $d(D, N)=DN\leq 0.07268$,  then we have
\begin{eqnarray*}
&(AB+BC+CA)+AD+BD+CD\leq 3\sqrt{3}+(AN+DN)+(BN+DN)
&\\
&+(CN+DN)\leq 3\sqrt{3}+3\sqrt{2}+3\times 0.07268 \leq 9.4388+0.21804<9.65685, 
&
\end{eqnarray*} 
and therefore,  $\{A, B, C, D\}$ is not the optimal configuration. 

Considering that the optimal configuration is invariant under the rotation around the $z$-axis,  
we can express the optimal problem to a  
non-linear programming as follows:
\begin{eqnarray}\label{four-point-problem}
\max && f:=\sum_{0\leq i<j\leq 3}\sqrt{(x_i-x_j)^2+(y_i-y_j)^2+(z_i-z_j)^2}, \\
\nonumber \mbox{s.t.}&& x_i^2+y_i^2+z_i^2=1,  \; i=0, 1, 2, 3, \\
\nonumber &&x_0=0, y_0=-1, z_0=z_1=z_2=0,  \; z_3\geq 0.
\end{eqnarray}
We may try to solve this problem by using the Lagrangian multiplier method and symbolic computer algebra. 
It is easy to prove that 
$$
A:=(0, -1, 0), \;
B:=(1, 0, 0), \;
C:=(-1, 0, 0), \;
D:=(0, 1, 0)
$$
form a local optimal solution of Problem~\eqref{four-point-problem},  but it is hard to prove that it is also a global maximum.  
The attempt for proving that $(A, B, C, D)$ is the unique critical point of $f: (S^1)^2\times S^2_{\geq 0}\rightarrow R$ 
was not successful since too large objects occurred in the elimination process. 
We have also applied several numerical algorithms to search the optimal configuration 
and found from experiments that $4\sqrt{2}+4$ seems to be the real global maximum. 
In this paper, we devote ourselves to construct a proof of this fact by combining numerical and symbolic computation. 
Our strategy is as follows. In the first stage (the {\em numerical global search} stage) we prove that
\begin{theorem}\label{th1}
If $A, B, C\in S^1$ and $D\in S^2_{\geq 0}$ form an optimal configuration for Problem \eqref{four-point-problem},  
then,  up to a rotation of hemisphere around the $z$-axis,  we have
$$
A=(0, -1, 0),  \quad B\in U,  \quad C\in V,  \quad D\in W
$$
where
\begin{eqnarray}
&&
U:=\{(x, y, 0)|1-\delta_1<x\leq 1, \; -\delta_2<y<\delta_2\}\\
&&
V:=\{(x, y, 0)|-1\leq x<-1+\delta_1, \; -\delta_2<y<\delta_2\}\\
&&
W:=\{(x, y, z)|-\delta_2<x<\delta_2,  \; 1-\delta_1<y\leq 1,  \; 0\leq z<\delta_2\} 
\end{eqnarray}
and $\delta_1=1/32,  \delta_2=1/4$. 
\end{theorem}

\noindent In this stage,  we divide the set of the feasible points of  Problem \eqref{four-point-problem},  namely, 
$(S_1)^2\times S^2_{\geq 0}\setminus (U\times V\times W)$,   
 into a disjoint union of finitely many small cubes and check the conjecture on each cube 
through estimating the upper bound of the objective function on that cube
with exact numerical computation of computer algebra software like {\sc Maple}, 
throw away those cubes where the conjecture is proved,  
and do branch-and-bound process recursively on the remained cubes. 
According to the Borel-Lebesgue covering theorem and the continuity of the objective function, 
this process will be terminated in finitely many rounds if the strict inequality 
$f<4\sqrt{2}+4$ is actually true on the consideration set. 

In the second stage (the {\em local critical analysis} stage),  we use symbolic computation to prove the following inequality:
\begin{theorem}\label{th2}
If $U, V, W$ are defined in Theorem~\ref{th1}. Then for any $B\in U,  C\in V,  D\in W$,  the following inequality is valid:
$$
AB+BC+CA+DA+DB+DC\leq 4+4\sqrt{2}.
$$
\end{theorem} 

\noindent In this stage,  we first verify that $A=(0, -1, 0)$,  $B=(1, 0, 0), C=(-1, 0, 0), D=(0, 1, 0)$ constitute a critical point of $f$,  and the objective
function can be expressed in the form
$$
f(A, B, C, D)=4\sqrt{2}+4+\frac{1}{Q(s, t, u, v)}\left(X\cdot H\cdot X^T+\mbox{h.o.t.}\right), 
$$
where $(s, t, u, v)\in {\Bbb R}^4$ are determined by
\begin{equation*}\label{coordinates-bc}
B
=(\frac{1-s^2}{1+s^2}, \frac{2s}{1+s^2}, 0),  \quad
C
=(-\frac{1-t^2}{1+t^2}, \frac{2t}{1+t^2}, 0), 
\end{equation*}
and
\begin{equation*}\label{coordinates-d}
D
=(\frac{2u}{1+u^2}, \;\frac{1-v^2}{1+v^2}, \;z_3)\in S^2_{\geq 0}, 
\end{equation*}
$Q(s, t, u, v)$ is a positive definite polynomial,  $H$ is a negative semi-definite $4\times 4$ symmetric matrix,  
and $\mbox{h.o.t.}$ stands for higher order terms of polynomials,  
then construct a symmetric matrix $M$,  in a mechanical way,  so that  
$$
-XMX^T \leq h.o.t. \leq XMX^T,  
$$
and $H+\rho\cdot M$ is negative semi-definite when $(s, t, u, v)$ satisfies that $(B, C, D)\in U\times V\times W$,  and $-1\leq \rho\leq 1$. 

Certainly,  by combining Theorem~\ref{th1} and Theorem~\ref{th2} together,  we get a complete solution to Problem \eqref{four-point-problem}.
It is also clear that if we can construct larger neighborhoods in the local critical analysis stage,  
then the computation will be reduced significantly in the global numerical search stage,  
since $f(X)$ is changing very slowly when $X$ approaches critical points.    
The first work to use this two-stage method for proving geometric inequality can be traced to 1988  when Jingzhong Zhang 
of Chengdu Branch of the Chinese Academy of Sciences gave a machine proof (unpublished)
on a Sharp PC-1500 
pocket computer 
to a geometric inequality of Zirakzadeh \cite{zirakzadeh64}, 
which says that given a triangle $ABC$,  
any points $P, Q, R$ on the boundary of $ABC$ which divide the perimeter of $ABC$  
into three equal lengths satisfy $PQ+QR+RP\geq (AB+BC+CA)/2$. A detailed explanation of Zhang's method
is given in \cite{zeng-zhang}. 
Related to the automated deduction in the local critical analysis stage,  a general construction method for computing the size of the locally optimal regions
of multivariable homogeneous polynomials is presented in \cite{zeng-yang11}. 
Notice also that in \cite{hou-shao11} Hou et al. reported a solution for maximizing the distance sum between five points  
on the unit sphere essentially using the two-stage method.

For the sake of page limitation,  we will concentrate mainly to make an explanation
to the local critical analysis (stage 2) for Problem~\eqref{four-point-problem} in this paper.
The paper is organized as follows. 
In Section~\ref{new-prelimnaries} we transform Problem~\eqref{four-point-problem} to an equivalent non-linear optimization problem
which has relatively simple objective function ${f}: (S^1)^2\times D^2\rightarrow R$.
In Section~\ref{how-large-is-small-neighbourhood} we prove a stronger version of Theorem~\ref{th2}.
In Section~\ref{gobal-numerical-search} we concisely describe the process of global numerical search for proving Theorem~\ref{th1}.

\section{Prelimnaries}
\label{new-prelimnaries}

\begin{lemma}\label{lemma0}
If  $A, B, C, D$ are four points on the upper hemisphere 
$S^2_{\geq 0}$ such that the sum of distances between them is maximal,   
then the sphere center $O$ is contained in the interior of $ABCD$  or lies on one surface of $ABCD$.
\end{lemma}

\begin{proof}
It is easy to see that the convex hull of any four points on a sphere is either a tetrahedron or a planar (convex) quadrilateral,  
and in the latter case,  if the sum of distances between them is maximal,  then they are all on the equator 
and $O$ is contained in the interior of the quadrilateral. 

Now we assume that the convex of $A,  B,  C,  D$ is a tetrahedron and that the center of the sphere $O$ lies in the outside of $ABCD$,  
then,  without loss of generality,  we can assume further 
that $D$ and $O$ lie at different sides of the plane $ABC$.
Let $r$ be the radius of the circumcircle of the triangle $ABC$,  $O'$ the center of the circle. 

If $r<1$,  then the plane determined by $A, B, C$ cuts the unit sphere 
into two spherical caps,  whereas $D$ lies on the smaller one (denoted by $K_D$),  
and $O$ lies in the interior of the convex hull of the larger one (denoted by $K_O$).
Construct the sphere ${S'}^2$ with center at $O'$ and radius $r$. 
Then $A,  B, C$ are on a great circle on ${S'}^2$ and the plane $ABC$ cuts ${S'}^2$ into two hemispheres.  
Let ${S'}^2_{\geq 0}$ be the hemisphere of ${S'}^2$ which lies in the same side of the plane $ABC$ with $D$.
Then,  $K_D$ is contained in the interior of ${S'}^2_{\geq 0}$ and 
$D$ is contained in the interior of the convex hull of ${S'}^2_{\geq 0}$. 
Let $D'$ be the point on ${S'}^2_{\geq 0}$ so that $D'D\perp ABC$. 
Then it is easy to see that 
$$
AD<AD',  BD<BD',  CD<CD',  
$$
and therefore,  
\begin{align*}
&AB+BC+CA+AD+BD+CD
< AB+BC+CA+AD'+BD'+CD'
\\
\leq& \max_{P'\in {S'}^{2}_{\geq 0}} (AB+BC+CA+AP'+BP'+CP')
\\
\leq &\max_{A_1, B_1, C_1, D_1\in {S'}^{2}_{\geq 0} }(A_1B_1+B_1C_1+C_1A_1+A_1D_1+B_1D_1+C_1D_2)
\\
=&\;r\cdot(AB+AC+AD+BC+CD+DB), 
\end{align*}  
which is impossible since $r<1$. 

If $r=1$ and $O$ is neither in the interior of $ABCD$ and nor in the triangle $ABC$ and its edges,  then 
$A, B, C$ are on the equator of $S^2_{\geq 0}$. 
Without loss of generality,  
we may assume that $A, O$ lie on different sides of line $BC$. Let $r_2$ be the radius of the circumcircle of the triangle $BCD$,  
$O''$ the center of the circle. Construct the sphere $S''^{2}$ with center at $O''$ and radius $r_2<1$. Then the plane $BCD$
cuts $S''^{2}$ into two hemispheres. Let $S''^{2}_{\geq 0}$ be the hemisphere that lies on the same side of the plane $BCD$ with $A$. 
Then $A$ is contained in the interior of $S''^{2}_{\geq 0}$. Let $A'\in S''^{2}_{\geq 0}$ be the point satisfying $A'A\perp BCD$. 
Then we will have
\begin{align*}
&AB+BC+CA+AD+BD+CD<A'B+BC+CA'+A'D+BD+CD
\\
\leq& r_2(AB+BC+CA+AD+BD+CD), 
\end{align*}
which is also impossible since $r_2<1$. 

Finally we proved that if $A,  B,  C,  D\in S^2_{\geq 0}$ form an optimal solution of Problem~\eqref{four-point-problem},  then 
$A,  B,  C$ are on the equator of the hemisphere,  and $O$ is in the interior or edges of $ABC$,  up to a permutation of the four points.   
\end{proof}

According to Lemma~\ref{lemma0},  we may assume that $A=(x_0, y_0, 0)=(0, -1,  0)$,  $B=(x_1, y_1, 0), C=(x_2, y_2, 0)\in S^1$ and $D=(x_3, y_3, z_3)\in S^2_{\geq 0}$
and search the maximal sum of the distances between the four points. 
It is clear that for any two points $P=(x, y, z),  Q=(x', y', z')\in S^2$ with $z\cdot z'=0$,  we have
$$
PQ=d(P, Q)=\sqrt{(x-x')^2+(y-y')^2+(z-z')^2}=\sqrt{2\cdot (1-x\, x'-y\,  y')}.
$$
Define the ``{\it warp distance\/}'' between two points $P_1=(x, y), Q_1=(x', y')$
in the unit disk $D^2$ as follows:
$$
\delta(P_1, Q_1):=\sqrt{2-2xx'-2yy'}.
$$
Let 
\begin{equation}\label{z-projection}
(\, )_1:{\Bbb R}^3\rightarrow {\Bbb R}^2,  \quad (x, y, z)\mapsto (x, y), 
\end{equation}
be the $z$-projection. Then for $A, B, C$ on the equator and $D$ on the hemisphere $S^2_{\geq 0}$,  we have
$$
d(P, Q)=\delta(P_1, Q_1)
$$
for $P, Q\in \{A, B, C, D\}$. 
Thus,  we can transform Problem~\eqref{four-point-problem} into the following optimization problem on $S^1\times S^1\times D^2$:
\begin{eqnarray}\label{warp-distance-programming}
\nonumber
\max&&\delta(A_1, B_1)+\delta(B_1, C_1)+\delta(C_1, A_1)+\delta(A_1, D_1)+\delta(B_1, D_1)+\delta(C_1, D_1), \\
\nonumber\mbox{s.t.}&&A_1=(0, -1), B_1=(x_1, y_1), C_1=(x_2, y_2)\in S^1:=\{(x, y)|x^2+y^2=1\}, 
\\
&&D_1=(x_3, y_3)\in D^2=\{(x, y)\;|\;x^2+y^2\leq 1\}.
\end{eqnarray}

It is clear that $\{A, B, C, D\}$ is an optimal configuration of Problem~\eqref{four-point-problem} 
if and only if $\{A_1, B_1, C_1, D_1\}$ is an optimal configuration of Problem~\eqref{warp-distance-programming}.  
We will prove the following result:

\begin{theorem}\label{morse-lemma}
If $A_1=(0,  -1)$ and $B_1=(x_1,  y_1), C_1=(x_2,  y_2), D_1=(x_3,  y_3)\in [-1,  1]\times [-1,  1]$
satisfy the following conditions:
\begin{align*}
(\mbox{\rm i})&\phantom{x} x_1^2+y_1^2=1, \; |x_1-1|\leq 1/25, \; |y_1|\leq 7/25, \\
(\mbox{\rm ii})&\phantom{x} x_2^2+y_2^2=1, \; |x_2+1|\leq 1/25, \; |y_2|\leq 7/25, \\
(\mbox{\rm iii})&\phantom{x} |x_3|\leq 7/25, \; |y_3-1|\leq 1/25, \\
(\mbox{\rm iv})&\phantom{x} \mbox{the triangle $ABC$ contains $O=(0, 0)$ in its inside or edges}, 
\end{align*}
then
$$
\delta(A_1, B_1)\!+\!\delta(B_1, C_1)\!+\!\delta(C_1, D_1)\!+\!\delta(A_1, D_1)\!+\!\delta(B_1, D_1)\!+\!\delta(C_1, D_1)\leq 
4\!+\!4\sqrt{2}, 
$$
and the equality holds if and only if $B_1=(1, 0), C_1=(-1, 0), D_1=(0, 1)$.
\end{theorem}

This result is stronger than Theorem~\ref{th2},  since if $A=(0,  -1,  0)$,  and
$$
B\in [\frac{31}{32}, 1]\times [-\frac{1}{4}, \frac{1}{4}],  C\in [-1, -\frac{31}{32}]\times [-\frac{1}{4}, \frac{1}{4}], 
D\in [-\frac{1}{4}, \frac{1}{4}]\times [\frac{31}{32}, 1]\times [0, \frac{1}{4}], 
$$
then their $z$-projections $A_1, B_1, C_1, D_1$ satisfy the conditions (i),  (ii) and (iii) in Theorem~\ref{morse-lemma}.
To prove Theorem~\ref{morse-lemma} we first setup a parametric representation for the coordinates of points $B_1, C_1, D_1$
in Problem~\eqref{warp-distance-programming}.
Notice that the property $O\in ABC$ in Lemma~\ref{lemma0} implies that $x_1x_2<0$ so we may assume that 
$$
x_1=\frac{1-s^2}{1+s^2},  \; y_1=\frac{2s}{1+s^2},  \quad
x_2=-\frac{1-t^2}{1+t^2},  \; y_2=\frac{2t}{1+t^2}, 
$$
for $-1\leq s, t\leq 1$,  and since the optimal configuration $A, B, C, D$ satisfies 
$$
AD+BD+CD\geq 4+4\sqrt{2}-(AB+BC+CA)\geq 4+4\sqrt{2}-3\sqrt{3}>3\sqrt{2}, 
$$
so without loss of generality, we may take $A=(0, -1, 0)$ such that 
$$
AD=\sqrt{2+2y_3}=\max\{AD, BD, CD\}>\sqrt{2}, 
$$ 
and therefore $y_3>0$,  and we assumed that
$$
x_3=\frac{2u}{1+u^2}, \;y_3=\frac{1-v^2}{1+v^2}, 
$$ 
for $-1\leq u\leq 1$ and $0\leq u\leq 1$. 
Furthermore,  the condition
$(x_3, y_3)\in D^2$ together with $v\geq 0$ implies that $-v\leq u\leq v$.
Thus,  the constraint conditions
of Problem~\eqref{warp-distance-programming} can be expressed as
\begin{align*}
&B_1=(\frac{1-s^2}{1+s^2}, \frac{2s}{1+s^2}), C_1=(-\frac{1-t^2}{1+t^2}, \frac{2t}{1+t^2}),  
D_1=(\frac{2u}{1+u^2}, \;\frac{1-v^2}{1+v^2}), 
\\
&-1\leq s, t, u\leq 1,  0\leq v\leq 1,  -v\leq u\leq v.
\end{align*}


We also need the following two lemmas in the proof of Theorem~\ref{morse-lemma}.
\begin{lemma}\label{taylor}
For $-1\leq x\leq 1$, 
$$
\sqrt{1-x}\leq 
1-\frac{1}{2}\, x-\frac{1}{8}\, x^2-\frac{1}{16}\, x^3.
$$ 
\end{lemma}

\begin{proof}
$$
\left(1-\frac{1}{2}\, x-\frac{1}{8}\, x^2-\frac{1}{16}\, x^3\right)^2
-(1-x)
={\frac {{x}^{4} \left( {x}^{2}+4\, x+20 \right) }{256}}>0.
$$
\end{proof}

\begin{lemma}\label{inequality-for-morse-lemma}
If $x_1, x_2, \cdots, x_n$ are positive real numbers such that $x_1, x_2, \cdots, x_n\leq r$,  then
\begin{equation}\label{ifml}
x_1^{d_1}x_2^{d_2}\cdots x_n^{d_n}\leq 
\frac{r^{N-2}}{N}\left(x_1^2\cdot d_1+x_2^2\cdot d_2+\cdots+x_n^2\cdot d_n\right), 
\end{equation}
where $N=d_1+d_2+\cdots+d_n$.
\end{lemma}

\begin{proof} 
Let $y_1, y_2, \cdots, y_N$ be a permutation of the following $N=d_1+d_2+\cdots+d_N$ numbers
$$
\overbrace{x_1, \cdots, x_1}^{d_1}, \;
\overbrace{x_2, \cdots, x_2}^{d_2}, \;
\cdots, \;
\overbrace{x_n, \cdots, x_n}^{d_n}, 
$$
then $0<y_1, y_2, \cdots, y_N\leq r$ and 
\begin{align*}
&\phantom{=}x_1^{d_1}x_2^{d_2}\cdots x_n^{d_n}
=y_1y_2\cdots y_N
\leq\frac{1}{N(N-1)}
\sum_{1\leq i<j<N} 2y_iy_j\cdot r^{N-2}
\\
&=\frac{r^{N-2}}{N(N-1)}
\left(
(y_1+y_2+\cdots+y_N)^2-(y_1^2+y_2^2+\cdots+y_N^2)
\right)
\\
&=\frac{r^{N-2}}{N(N-1)}
\left(
(d_1x_1+d_2x_2+d_nx_n)^2-(d_1x_1^2+d_2x_2^2+\cdots+d_nx_n^2)
\right)
\\
&=\frac{r^{N-2}}{N(N-1)}
\left(
\sum_{k=1}^{n}(d_k^2-d_k)x_k^2\;
+\;2\sum_{1\leq i<j\leq n}d_id_j\cdot x_ix_j 
\right)
\\
&\leq\frac{r^{N-2}}{N(N-1)}
\left(
\sum_{k=1}^{n}(d_k^2-d_k)x_k^2\;
+\sum_{1\leq i<j\leq n} d_id_j(x_i^2+x_j^2)
\right).
\end{align*}
Now we compute 
$$
S=\sum_{1\leq i<j\leq n} d_id_j(x_i^2+x_j^2)
$$
as follows.
\begin{align*}
2S
&
=\sum_{1\leq i, j\leq n} d_id_j(x_i^2+x_j^2)
-\sum_{k=1}^{n} d_k^2(x_k^2+x_k^2)
\\
&=\sum_{i=1}^n\Big[d_ix_i^2(\sum_{j=1}^nd_j)\Big]+\sum_{j=1}^n\Big[(\sum_{i=1}^nd_i)d_jx_j^2\Big]
-2\sum_{k=1}^{n} \big[d_k^2x_k^2\big]
\\
&=2N \sum_{i=1}^nd_ix_i^2-2\sum_{k=1}^{n} d_k^2x_k^2.
\end{align*} 
Therefore, 
$$
x_1^{d_1}x_2^{d_2}\cdots x_n^{d_n}
\leq
\frac{r^{N-2}}{N(N\!-\!1)}\!
\cdot\! {(N\!-\!1)}\sum_{k=1}^{n}d_kx_k^2
=\frac{r^{N-2}}{N}(d_1x_1^2+d_2x_2^2+\cdots+d_nx_n^2), 
$$
as claimed. 
\end{proof}


\section{Automated Local Critical Analysis}
\label{how-large-is-small-neighbourhood}

Now we give the proof of Theorem~\ref{morse-lemma}.

\begin{proof} For simplicity we drop the subscripts of $A_1, B_1, C_1, D_1$. Assume that
$$
B=(x_1, y_1)=(\frac{1-s^2}{1+s^2}, \frac{2s}{1+s^2}),  \;
C=(x_2, y_2)=(-\frac{1-t^2}{1+t^2}, \frac{2t}{1+t^2}), 
$$
and
$$
D=(x_3, y_3)=(\frac{2u}{1+u^2}, \frac{1-v^2}{1+v^2}).
$$
Then,  the conditions (i),  (ii),  (iii) in Theorem~\ref{morse-lemma}
can be transformed into 
$$
-1/7\leq s, t, u, v\leq 1/7, 
$$
and the condition (iv) that $O=(0, 0)$ is in the inside (or on the edges) of $ABC$ can be represented by 
that the oriented area of $ABC$ is positive,  
\begin{equation}\label{splustispositive}
\frac{1}{2}\, {\tt det}  \left[ \begin {array}{ccc} 
0&0&1\\ 
\noalign{\medskip}
{\frac {-{s}^{2}+1}{{s}^{2}+1}}&{\frac {2\, s}{{s}^{2}+1}}&1\\ 
\noalign{\medskip}
-{\frac {-{t}^{2}+1}{{t}^{2}+1}}&{\frac {2\, t}{{t}^{2}+1}}&1\end {array}
 \right] ={\frac { \left( s+t \right)  \left( 1-s\, t \right) }{
 \left( {t}^{2}+1 \right)  \left( {s}^{2}+1 \right) }}
\geq 0, \;
\mbox{ i.e.,  }\; s+t\geq 0. 
\end{equation}

\noindent
Applying Lemma~\ref{taylor},  we have
\begin{align}
\label{dw12}
\delta(A, B)&=\sqrt {2+{\frac {4\, s}{{s}^{2}+1}}}
\leq
\sqrt {2}+{\frac {\sqrt {2}s}{{s}^{2}+1}}-{\frac {\sqrt {2}{s}^{2
}}{ 2\, \left( {s}^{2}+1 \right) ^{2}}}+{\frac {\sqrt {2}{s}^{3}}{
2\,  \left( {s}^{2}+1 \right) ^{3}}}, 
\\[3pt]
\label{dw13}
\delta(A, C)&=\sqrt {2+{\frac {4\, t}{{t}^{2}+1}}}
\leq
\sqrt {2}+{\frac {\sqrt {2}t}{{t}^{2}+1}}-{\frac {\sqrt {2}{t}^{2
}}{ 2\, \left( {t}^{2}+1 \right) ^{2}}}+{\frac {\sqrt {2}{t}^{3}}{
 2\, \left( {t}^{2}+1 \right) ^{3}}}, 
\\[3pt]
\label{dw14}
\delta(A, D)&=\sqrt {2+2\, {\frac {-{v}^{2}+1}{{v}^{2}+1}}}
\leq
2-{\frac {{v}^{2}}{{v}^{2}+1}}-{\frac {{v}^{4}}{ 4\,  \left( {v}^{2}+1
 \right) ^{2}}}-{\frac {{v}^{6}}{ 8\, \left( {v}^{2}+1 \right) ^{3}}}, 
\\[5pt]
\delta(B, C)&=
\sqrt {2+2\, {\frac { \left( -{t}^{2}+1 \right)  \left( -{s}^{2}+1
 \right) }{ \left( {t}^{2}+1 \right)  \left( {s}^{2}+1 \right) }}-{
\frac {8\, st}{ \left( {t}^{2}+1 \right)  \left( {s}^{2}+1 \right) }}}
\nonumber\\[3pt]
&=
2\, \sqrt {1-{\frac {{s}^{2}+{t}^{2}}{ \left( {t}^{2}+1 \right) 
 \left( {s}^{2}+1 \right) }}-{\frac {2\, st}{ \left( {t}^{2}+1 \right) 
 \left( {s}^{2}+1 \right) }}}
\nonumber\\[3pt]
&\label{dw23}
\leq
2\!-\!{\frac {{s}^{2}+2\, st+{t}^{2}}{ \left( {t}^{2}+1 \right)  \left( {s}
^{2}+1 \right) }}\!-\!{\frac { \left( {s}^{2}+2\, st+{t}^{2} \right) ^
{2}}{ 4\, \left( {t}^{2}+1 \right) ^{2} \left( {s}^{2}+1 \right) ^{2}}}\!-\!
{\frac { \left( {s}^{2}+2\, st+{t}^{2} \right) ^{3}}{ 8\, \left( {t}^{2}
+1 \right) ^{3} \left( {s}^{2}+1 \right) ^{3}}}, 
\end{align}
\begin{align}
\delta(D, B)
&=
\sqrt {2+4\, {\frac { \left( -{s}^{2}+1 \right) u}{ \left( {s}^{2}+1
 \right)  \left( {u}^{2}+1 \right) }}-4\, {\frac {s \left( -{v}^{2}+1
 \right) }{ \left( {s}^{2}+1 \right)  \left( {v}^{2}+1 \right) }}}
\nonumber
\\[3pt]
&\leq
\sqrt {2}-{\frac {\sqrt {2} \left( {s}^{2}u{v}^{2}-s{u}^{2}{v}^{2}+{s}
^{2}u+s{u}^{2}-s{v}^{2}-u{v}^{2}+s-u \right) }{ \left( {s}^{2}+1
 \right)  \left( {u}^{2}+1 \right)  \left( {v}^{2}+1 \right) }}
\nonumber\\[3pt]
&\phantom{\leq}
-1/2\, {
\frac {\sqrt {2} \left( {s}^{2}u{v}^{2}-s{u}^{2}{v}^{2}+{s}^{2}u+s{u}^
{2}-s{v}^{2}-u{v}^{2}+s-u \right) ^{2}}{ \left( {s}^{2}+1 \right) ^{2}
 \left( {u}^{2}+1 \right) ^{2} \left( {v}^{2}+1 \right) ^{2}}}
\nonumber\\[3pt]
&\phantom{\leq}
-1/2\, {
\frac {\sqrt {2} \left( {s}^{2}u{v}^{2}-s{u}^{2}{v}^{2}+{s}^{2}u+s{u}^
{2}-s{v}^{2}-u{v}^{2}+s-u \right) ^{3}}{ \left( {s}^{2}+1 \right) ^{3}
 \left( {u}^{2}+1 \right) ^{3} \left( {v}^{2}+1 \right) ^{3}}},  \label{dw34}
\end{align}
\begin{align}
\delta(C, D)
&=
\sqrt {2-4\, {\frac { \left( -{t}^{2}+1 \right) u}{ \left( {t}^{2}+1
 \right)  \left( {u}^{2}+1 \right) }}-4\, {\frac {t \left( -{v}^{2}+1
 \right) }{ \left( {t}^{2}+1 \right)  \left( {v}^{2}+1 \right) }}}
\nonumber\\[3pt]
&\phantom{\leq}
\leq
\sqrt {2}+{\frac {\sqrt {2} \left( {t}^{2}u{v}^{2}+t{u}^{2}{v}^{2}+{t}
^{2}u-t{u}^{2}+t{v}^{2}-u{v}^{2}-t-u \right) }{ \left( {t}^{2}+1
 \right)  \left( {u}^{2}+1 \right)  \left( {v}^{2}+1 \right) }}
\nonumber\\[3pt]
&\phantom{\leq}
-1/2\, {
\frac {\sqrt {2} \left( {t}^{2}u{v}^{2}+t{u}^{2}{v}^{2}+{t}^{2}u-t{u}^
{2}+t{v}^{2}-u{v}^{2}-t-u \right) ^{2}}{ \left( {t}^{2}+1 \right) ^{2}
 \left( {u}^{2}+1 \right) ^{2} \left( {v}^{2}+1 \right) ^{2}}}
\nonumber\\[3pt]
&\phantom{\leq}
+1/2\, {
\frac {\sqrt {2} \left( {t}^{2}u{v}^{2}+t{u}^{2}{v}^{2}+{t}^{2}u-t{u}^
{2}+t{v}^{2}-u{v}^{2}-t-u \right) ^{3}}{ \left( {t}^{2}+1 \right) ^{3}
 \left( {u}^{2}+1 \right) ^{3} \left( {v}^{2}+1 \right) ^{3}}}.\label{dw24}
\end{align}
Summer up the right sides of \eqref{dw12} to \eqref{dw24},  we have
\begin{align*}
&\delta(A, B)+\delta(A, C)+\delta(A, D)+\delta(B, C)+\delta(C, D)+\delta(D, B)
\\[3pt]
\leq& 
4+4\sqrt{2}
+
\frac{J(s, t, u, v)}
{8\,  \left( {t}^{2}+1 \right) ^{3} \left( {s}^{2}+1 \right) ^{3}
 \left( {v}^{2}+1 \right) ^{3} \left( {u}^{2}+1 \right) ^{3}}
, 
\end{align*}
where $J(s, t, u, v)$ is a polynomial of $s, t, u, v$ with $1288$ monomials,  and the lowerest and higher degree of the monomials are $2$ and $24$,  respectively. Let $H_j$ be the homogeneous terms of degree $j$. To prove Theorem~\ref{morse-lemma},  we need to verify that $J(s, t, u, v)\leq 0$ for all $s, t, u, v\in [-1/7, 1/7]$ with extra condition $s+t\geq 0$.  

The the expanded form of $H_2, H_3, H_4$ has $9, 20, 28$ terms,  respectively,  
\begin{align}
H_2=&
-8( \sqrt {2}+1) {s}^{2}-16st-8
 ( \sqrt {2}+1 ) {t}^{2}+8\sqrt {2}su-8\sqrt {2}tu-8\sqrt {2}{u}^{2
}-8{v}^{2}, 
\label{xph2}
\\
H_3=&
-4\, \sqrt {2} \left( su-tu+3\, {u}^{2}-4\, {v}^{2} \right)  \left( s+t
 \right), 
\label{xph3}
\\
H_4=&
-18{v}^{4}-8\sqrt {2}{s}^{4}-8\sqrt {2}{t}^{4}-8\sqrt {2}{u}^{
4}-48\sqrt {2}{s}^{2}{t}^{2}-32\sqrt {2}{s}^{2}{u}^{2}-8\sqrt {2
}{s}^{2}{v}^{2}
\nonumber
\\
&
+16\, \sqrt {2}s{u}^{3}-32\, \sqrt {2}{t}^{2}{u}^{2}-8\, 
\sqrt {2}{t}^{2}{v}^{2}-16\, \sqrt {2}t{u}^{3}-24\, \sqrt {2}{u}^{2}{v}^
{2}-44\, {s}^{2}{t}^{2}
\nonumber
\\
&
-24\, {s}^{2}{u}^{2}-48\, {s}^{2}{v}^{2}-24\, {t}^{
2}{u}^{2}-48\, {t}^{2}{v}^{2}-24\, {u}^{2}{v}^{2}-40\, {s}^{3}t-40\, s{t}^
{3}
\nonumber
\\
&
-48\, st{u}^{2}-48\, st{v}^{2}-24\, \sqrt {2}{s}^{2}tu+24\, \sqrt {2}s{
t}^{2}u+8\, \sqrt {2}su{v}^{2}
\nonumber
\\
&-8\, \sqrt {2}tu{v}^{2}-18\, {s}^{4}-18\, {t
}^{4}.
\label{xph4}
\end{align}
The number of monomials in $H_5, H_6,  \cdots,  H_{24}$ are listed as follows:
$$
20, 59, 44, 101, 70, 134, 88, 145, 90, 133, 74, 100, 50, 59, 26, 29, 10, 10, 2, 1, 
$$
and
$$
H_{23}=16\, \sqrt {2}{s}^{6}{t}^{5}{u}^{6}{v}^{6}+16\, \sqrt {2}{s}^{5}{t}^{6}{
u}^{6}{v}^{6}, 
\quad
H_{24}=-11\, {s}^{6}{t}^{6}{u}^{6}{v}^{6}.
$$
Let $H:=H_5+H_6+\cdots+H_{24}$. Let 
$$
s'=\abs(s), \; t'=\abs(t), \; u'=\abs(u),  v'=\abs(v), 
$$
be the absolute values and the transformation
$$
{\mathcal T}: {\Bbb R}[s, t, u, v]
\longrightarrow
{\Bbb R}[s', t', u', v']
$$
be defined by
$$
{\mathcal T} \left(\sum a_{d_1, d_2, d_3, d_4} s^{d_1}t^{d_2}u^{d_3}v^{d_4}\right)
=
\sum b_{d_1, d_2, d_3, d_4} {s'}^{d_1}{t'}^{d_2}{u'}^{d_3}{v'}^{d_4},
$$
where
$$
b_{d_1, d_2, d_3, d_4}
=\left\{
\begin{array}{l}
0, 
\mbox{ if all $d_1, d_2, d_3, d_4$ are even integers,  and $a_{d_1, d_2, d_3, d_4}<0$}, 
\\[3pt]
\abs(a_{d_1, d_2, d_3, d_4}), \mbox{ otherwise}.
\end{array}
\right.
$$
Then 
\begin{equation}\label{j2thee}
J(s, t, u, v)\leq H_2+H_3+H_4+\thee(s', t', u', v'), 
\end{equation}
here
$$\thee(s', t', u', v')={\mathcal T}(H_5)+\cdots+{\mathcal T}(H_{24})
$$ 
has 797 monomials,  in which ${\mathcal T}(H_{24})=0$,  and 
${\mathcal T}(H_d)\;(d=5, 6, \cdots, 23)$ has
$$
   20,  24,  44,  42,  70,  57,  88,  64,  90,  57,  74,  42,  50,  24,  26,  10,  10,  3,  2,  
$$
monomials,  respectively. It is clear that for odd $d$,  ${\mathcal T}(H_d)$ and $H_d$ have the equal number of monomials,  and for even $d$,  
the number of monomials in ${\mathcal T}(H_d)$ might have less than that in $H_d$. 

Now variables $(s', t', u', v')$ and coefficients of polynomials ${\mathcal T}(H_d)$ for $(d=5, 6, \cdots, 23)$ are all positive. 
We use the following transformation to map them to quadratic polynomials.
For each $d$,  we define a transformation over homogeneous polynomials of degree $d$ to a quadratic polynomial of $s, t, u, v$ as follows:
$$
{\mathcal S}_d
 \left(\sum_{d_1+d_2+d_3+d_4=d} 
b_{d_1, d_2, d_3, d_4} 
{s'}^{d_1}{t'}^{d_2}{u'}^{d_3}{v'}^{d_4}\right)
$$
$$
=
 \sum b_{d_1, d_2, d_3, d_4} \frac{r_0^{d-2}}{12}
\left(
(s\cdot{d'_1})^2+(t\cdot{d'_2})^2+(u\cdot{d'_3})^2+(v\cdot{d'_4})^2
\right).
$$
Here,  we take $r_0=1/7$,  and
$$
d'_1=4d_1^2-d_1,  \;
d'_2=4d_2^2-d_2,  \;
d'_3=4d_3^2-d_3,  \;
d'_4=4d_4^2-d_4,  
$$
Since we assumed that 
$-1\leq s, t, u, v\leq r_0=1/7$,  so $s', t', u', v'\in [0, 1/7]$,  and
the inequality
\begin{align*}
{s'}^{d_1}{t'}^{d_2}{u'}^{d_3}{v'}^{d_4}
&
\leq 
\frac{r_0^{d-2}}{4\times 3}
\cdot
\left(
{({d'_1}s')^2+({d'_2}t')^2+({d'_3}u')^2+({d'_4}v')^2}
\right)
\\
&
\leq 
\frac{r_0^{d-2}}{12}
\cdot
\left(
{(s\cdot{d'_1})^2+(t\cdot{d'_2})^2+(u\cdot{d'_3})^2+(v\cdot{d'_4})^2}
\right)
\end{align*}
is valid for all $d\geq 2$,  according to Lemma~\ref{inequality-for-morse-lemma}. Therefore,  we have following inequalities:
$$
{\mathcal T}(H_d)(s', t', u', v')\leq 
{\mathcal S_d}({\mathcal T}(H_d)),  \; d=5, \cdots, 23, 
$$
and
$$
{\mathcal T}(H)(s', t', u', v')\leq
\sum_{d=5, \cdots, 23} {\mathcal S_d}({\mathcal T}(H_d)):={\tt res5}.
$$
The computation of ${\tt res5}$ is as follows:
\begin{align*}
{\tt res5}
&={\frac {2223743956730603493021422\, {s}^{2}}{2198957644322995555530531}
}+{\frac {351460055057882361271126\, {u}^{2}}{377598787408999236808273}
}\\
&
+{\frac {39371575001649787465938178\, {v}^{2}}{
37382279953490924444019027}}+{\frac {2223743956730603493021422\, {t}^{2
}}{2198957644322995555530531}}
\\
&=1.01127\ldots {s}^2 + 0.93077\ldots {u}^2 + {1.05321\ldots}{v}^2 + 1.01127\ldots{t}^2.
\\
&\leq \frac{10}{9}s^2+\frac{10}{9}t^2+u^2+\frac{10}{9}v^2.
\end{align*}
Here the equality holds if and only if $s=t=u=v=0$. Thus,  the inequality \eqref{j2thee} implies that
\begin{equation}\label{j2res5}
J(s, t, u, v)\leq H_2+H_3+H_4+
\left(
\frac{5}{4}s^2+\frac{5}{4}t^2+u^2+\frac{5}{4}v^2
\right).
\end{equation}
Notice that from \eqref{xph2},  \eqref{xph3} and \eqref{xph4},  we have
\begin{align*}
H_2=
&-8\, v^2+k_{20}(s, t, u), 
\\
H_3=
&16\sqrt{2}(s+t)\, v^2+k_{30}(s, t, u), 
\\
H_4=
&-18\, v^4+k_{42}(s, t, u)\, v^2+k_{40}(s, t, u), 
\end{align*}
where $k_{20}, k_{30}, k_{42}, k_{40}$ are polynimials of $s, t, u$,  and
$$
k_{42}
=
-8\, \sqrt {2}{s}^{2}+8\, \sqrt {2}su-8\, \sqrt {2}{t}^{2}-8\, \sqrt {2}tu
-24\, \sqrt {2}{u}^{2}-48\, {s}^{2}-48\, st-48\, {t}^{2}-24\, {u}^{2}.
$$
Let $k_{22}=-8,  k_{32}=16\sqrt{2}{s+t}$ and
$$
K_2(s, t, u):=k_{22}+k_{32}+k_{42}+\frac{5}{4}.
$$
Then $K_2$ is a polynomial of $s, t, u$ of degree $2$,  and
\begin{align*}
\frac{\partial K_2}{\partial s}
&=-16\, \sqrt {2}s-96\, s-48\, t+8\, \sqrt {2}u+16\, \sqrt {2}, 
\\
\frac{\partial K_2}{\partial t}
&=-96\, t-48
\, s-8\, \sqrt {2}u-16\, \sqrt {2}t+16\, \sqrt {2}, 
\\
\frac{\partial K_2}{\partial u}
&=-48\, u+8\, \sqrt {2}
s-8\, \sqrt {2}t-48\, \sqrt {2}u.
\end{align*}
so its critical point is 
$$
\left( -{\frac{2}{79}}+{\frac {9\, \sqrt {2}}{79}}, -{\frac{2}{79}
}+{\frac {9\, \sqrt {2}}{79}}, 0 \right), 
$$
which is outside cube $[-1/7, 1/7]\times [-1/7, 1/7]\times [-1/7, 1/7]$,  which means that $K_2(s, t, u)$ has no local maximal or minimal point inside the cube. 
It is also verified that $K(0, 0, 0)=-62/9$,  and on each face of this cube,  $K(s, t, u)$ is also negative,  therefore,  we have 
$$
K_2(s, t, u)<0, 
$$ 
for all $s, t, u\in [-1/7, 1/7]$. Therefore,  we have the following inequality.
\begin{align}
J(s, t, u, v)&\leq -18v^4+K_2(s, t, u)v^2+ k_{20}+k_{30}+k_{40}+
\left(
\frac{5}{4}s^2+\frac{5}{4}t^2+u^2
\right)
\nonumber
\\
&
\leq k_{20}+k_{30}+k_{40}+
\left(
\frac{5}{4}s^2+\frac{5}{4}t^2+u^2
\right),  \label{j2stu}
\end{align}
here
\begin{align*}
k_{20}=&
-8\, \sqrt {2}{s}^{2}-8\, \sqrt {2}{t}^{2}-8\, \sqrt {2}{u}^{2}+8\, \sqrt 
{2}su-8\, \sqrt {2}tu-16\, st-8\, {t}^{2}-8\, {s}^{2}, 
\\
k_{30}=&
 -12\, \sqrt {2}\left( s+t \right) {u}^{2}+ 
4\, \sqrt {2}\left( -{s}^{2}+{t}^{2} \right) u
, 
\\
k_{40}=&
-8\, \sqrt {2}{s}^{4}-8\, \sqrt {2}{t}^{4}-8\, \sqrt {2}{u}^{4}-48\, 
\sqrt {2}{s}^{2}{t}^{2}-32\, \sqrt {2}{s}^{2}{u}^{2}
-32\, \sqrt {2}{t}^{2}{u}^{2}
\\
&
-24\, {s}^{2}{u}^{2}-24\, {t}^{2}{u}^{2}
-44\, {s}^{2}{t}^{2}
-18\, {s}^{4}-18\, {t}^{4}
\\
&
+16\, \sqrt {2}s{u}^{3}
-16\, \sqrt {2}t{u}^{3}-40\, {s}^{3}t-40\, s{t}^{3}
\\
&
-48\, st{u}^{2}-24\, \sqrt {2}{s}^{2}tu+24\, \sqrt {2}s{t}^{2}u.
\end{align*}
Applying Lemma~\ref{inequality-for-morse-lemma} we can get the following inequality for $k_{30}$. 
\begin{align}
k_{30}
&\leq
4\, \sqrt {2}\left( -{s}^{2}+{t}^{2} \right) u
\leq 
4\, \sqrt {2}{s'}^{2}u'+4\, \sqrt {2}{t'}^{2}u'
\nonumber\\
&\leq \frac{4}{3}\, \sqrt{2}(s'\cdot s'u'+ s'\cdot s'u'+ u'\cdot s' s')
+
\frac{4}{3}\, \sqrt{2}(t'\cdot t'u'+ t'\cdot t'u'+ u'\cdot t' t')
\nonumber\\
&\leq \frac{4}{21}\, \sqrt{2}(s'u'+ s'u'+ s' s')
+
\frac{4}{21}\, \sqrt{2}(t'u'+ t'u'+ t' t')
\nonumber\\
&\leq
\frac{2}{21}\, \sqrt {2} \left( 4\, {s}^{2}+2\, {u}^{2} \right) 
+
\frac{2}{21}\, \sqrt {2} \left( 4\, {t}^{2}+2\, {u}^{2} \right)
\nonumber\\
&=
{\frac {8\, \sqrt {2}{s}^{2}}{21}}+{\frac {8\, \sqrt {2}{u}^{2}}{21}}+{
\frac {8\, \sqrt {2}{t}^{2}}{21}}. \label{k30b}
\end{align}
For $k_{40}$,  
applying Lemma~\ref{inequality-for-morse-lemma} we get:
\begin{align}
k_{40}
&
\leq
16\, \sqrt {2}s{u}^{3}+16\, \sqrt {2}t{u}^{3}+40\, {s}^{3}t+40\, s{t}^{3}+
48\, st{u}^{2}+24\, \sqrt {2}{s}^{2}tu+24\, \sqrt {2}s{t}^{2}u
\nonumber\\
&\leq
\left( {\frac {22\, \sqrt {2}}{49}}+{\frac{52}{49}} \right) {s}^{2}+
 \left( {\frac {22\, \sqrt {2}}{49}}+{\frac{52}{49}} \right) {t}^{2}+
 \left( {\frac {36\, \sqrt {2}}{49}}+{\frac{24}{49}} \right) {u}^{2}.
\label{k40b}
\end{align}
Combining \eqref{j2stu}, \eqref{k30b}, \eqref{k40b},  we proved that under the assumption $-1/7<s, t, u<1/7$ and $s+t>0$, 
\begin{align}\label{j2k2345}
&J(s, t, u, v)
\leq
q_2(s, t, u):=k_{20}(s, t, u)
+
 \left( {\frac {122\, \sqrt {2}}{147}}+{\frac{453}{196}} \right) {s}^{2
}
\nonumber 
\\
&\phantom{q_2(s, t, u):=k_{20}(s, t, u)+x}
+ \left( {\frac {122\, \sqrt {2}}{147}}+{\frac{453}{196}} \right) {t}^
{2}+ \left( {\frac {164\, \sqrt {2}}{147}}+{\frac{73}{49}} \right) {u}^
{2}
\nonumber 
\\
&
=\frac{1}{2}(s, \;t, \;u)
 \left[ \begin {array}{ccc} -{\frac {2108\, \sqrt {2}}{147}}-{\frac{
1115}{98}}&-16&8\, \sqrt {2}\\ \noalign{\medskip}-16&-{\frac {2108\, 
\sqrt {2}}{147}}-{\frac{1115}{98}}&-8\, \sqrt {2}\\ \noalign{\medskip}8
\, \sqrt {2}&-8\, \sqrt {2}&-{\frac {2024\, \sqrt {2}}{147}}+{\frac{146}{
49}}\end {array} \right]
\left( \begin{array}{c}
s\\[5pt]t\\[5pt]u\end{array}
\right).
\end{align}
With computer algebra or hand computation,  it is easy now to check that
the matrix in \eqref{j2k2345}
is a negative semidefinite. Therefore,  the inequality
$$
J(s, t, u, v)\leq 0, 
$$
is valid for all $-1/7<s, t, u, v<1/7$ under the condition $s+t\geq 0$. This completes the proof of Theorem~\ref{morse-lemma}.
\end{proof}


\begin{remark}
The constant $r_0=1/7$ can be enlarged to $1/6.7845\approx 0.1473$.
\end{remark}

\section{Sketch of the Global Numerical Search}
\label{gobal-numerical-search}

In this section, we briefly describe the global numerical search process for proving Theorem~\ref{th1}. 
We need the following property of the optimal configuration. 
\begin{lemma}\label{lemma2}
Suppose that $A, B, C\in S^1$,  $D\in S^2_{\geq 0}$ and $AB+BC+CA+DA+DB+DC$ is maximal,  and
$d_1\leq d_2\leq \cdots\leq d_6$ is a permutation of $AB, BC, CA, DA, DB, DC$. Then
$$
d_1\geq 0.99200,  \;
d_2\geq 1.21895,  \;
d_3\geq 4/3, \;
d_4\geq \sqrt{2}, \;
d_5\geq 1.53137, \;
d_6\geq 1.60947.
$$
\end{lemma}
We will not show the proof of Lemma~\ref{lemma2} here because of space limitation. The global numerical search 
for Theorem~\ref{th1} is composed of three steps as follows. 

{\bf Step 1.} Dividing the unit square $[-1, 1]\times [-1, 1]\subset {\Bbb R}^2$ into $256$ squares of edge $1/8$,  
then we will see that $224$ of them have non-empty intersection with 
$D^2=\{(x, y)|x^2+y^2\leq 1\}$,  and $60$ of them have non-empty intersection with $S^1=\{(x, y)|x^2+y^2=1\}$.
From them,  we build a set of $N_1=60\times 60\times 224=806, 400$ cubes of edge $1/8$ (called $\frac{1}{8}$-cubes) in ${\Bbb R}^6$  
to cover the feasible set $(S^1)^2\times D^2$ of Problem~\eqref{warp-distance-programming}.
Applying Lemma~\ref{lemma2} to test the bounds of six distances (called the {\it distance bound test\/}) we
 can see that only  $11\times 11\times 88=10, 648$ are possible combinations
to locate the optimal solution of Problem~\ref{warp-distance-programming}; applying the exact numerical
computation to estimate the sum of the six distances (called the {\it distance sum test\/}) we see that
among the $10, 648$ $\frac{1}{8}$-cubes there are only $4, 300(=0.533\%\times N_1)$ need to be divided and checked in the next round. 
The first round checking has used $15.922$ seconds on a notebook computer with Intel COREi7 8th Gen CPU.

{\bf Step 2.} Partition each $\frac{1}{8}$-cube into $16\times 16\times 16=4, 096$ equal cubes of edge $1/32$ (called
$\frac{1}{32}$-cubes),  we 
get a set of $4, 300\times 4, 096=17, 612, 800$ cubes in ${\Bbb R}^6$,  
among them there are $N_2=1, 105, 782$ having non-empty intersection with $(S^1)^2\times S^1\times D^2$. 
Apply the {\it distance bound test\/}) we can verify that $844, 917$ are possible 
to locate the optimal solution.  From these suspect cubes,  remove $2, 048$ $\frac{1}{32}$-cubes that are 
contained $U\times V\times W_1$,  where $U, V, W$ are defined in Theorem~\ref{th1} and $W_1$ is the $z$-projection defined in \eqref{z-projection},  
and do the {\it distance sum test\/} so to remove another $823, 663$ $\frac{1}{32}$-cubes that are clearly non-optimal combination. 
Therefore,  there are $844, 917-2, 048-823, 663=19, 206(=1.737\%\times N_2)$ cubes of edge $1/32$ to be ckecked further.
Computation in this step has used $1, 170.266$ seconds. 

{\bf Step 3.} In the first two steps we do {\it breadth-first search\/} (BFS),  
in this step we do {\it deep-first search\/} (DFS) on each $\frac{1}{32}$-cube using only 
{\it distance sum test\/}. Namely,  we take one $\frac{1}{32}$ and divide it to $4, 096$ cubes of edge $1/128$ and estimate
the sum of distance,  if some of the $\frac{1}{128}$ have not passed the test,  we take first of them and divide it into $4, 096$ cubes of 
edge $1/512$,  and so on,  and return to parent level until all children-cubes passed the test.  
The computation has shown that among the $19, 206$ $\frac{1}{32}$-cubes,   $19, 107$ has passed the DFS checking on children-cubes of edge $1/128$,  
and the rest $199$ cubes passed on when the edge of children-cubes is $1/512$.  The $DFS$ computation has been completed in $8, 777.250$ seconds.

The exact numerical computation in this part and the symbolic computation in Section~\ref{how-large-is-small-neighbourhood}
are implemented with the computer algebra software Maple.  
We will publish the proofs of Lemma~\ref{lemma2} and Theorem~\ref{th1} in full in the Maple Conference 2021.


\providecommand{\urlalt}[2]{\href{#1}{#2}}
\providecommand{\doi}[1]{doi:\urlalt{http://dx.doi.org/#1}{#1}}

%
%
%
%

\end{document}